\documentclass{article}
\usepackage{amsmath,amssymb,amsthm}
\usepackage{type1cm}
\usepackage{ascmac}
\usepackage{graphicx}
\usepackage{xypic}
\usepackage{braket}
\usepackage{color}
\newtheorem{defi}{Definition}[section]
\newtheorem{theo}{Theorem}[section]
\newtheorem{lemma}{Lemma}[section]
\newtheorem{example}{Example}[section]
\newtheorem{cor}{Corollary}[section]
\title{Extension of the Sheaf-theoretic Structure to \\ Algebraic Quantum Field Theory}
\author{Tsubasa Takagi
\footnote{
e-mail: takagi.tsubasa283@gmail.com
}
 \\ Koukou at Otsuka, University of Tsukuba,\\
1-9-1 Otsuka, Tokyo, 112-0012, Japan}
\date{}
\begin{document}
\maketitle
\begin{abstract}
The sheaf-theoretic structure is useful in classifying no-go theorems related to non-locality and contextuality. It provides a new point of view different from conventional formularization of quantum mechanics. First, we examine a relationship between the conventional formularization and the innovative formularization. There exists an equivalence of their categories, and from the equivalence, one locality can be transformed to another as a concrete example. Next, we extend the quantum mechanics which has a finite-degree of freedom to  the quantum filed theory with an infinite-degree of freedom, especially to the algebraic quantum field theory (AQFT for short). We consider about a violation of the Bell inequality in AQFT, and we show that the condition of strict spacelike separation has the same Cartesian product structure as locality of quantum mechanics. Also, we show that no-signalling property can be proved by Split Property. A local state is a sheaf which is defined by Split Property in AQFT. It induces the sheaf-theoretic structure. Finally, we show an extension of the No-Signalling theorem which depends on spacetime in AQFT.
\end{abstract}

\section{Introduction}
According to quantum mechanics, a measured value of a physical system does not correspond to another measured value of the same physical system for other observation. We can only predict a probability distribution (more precisely, called a wave function). Now, we consider about a problem "Can any probabilistic phenomena of quantum mechanics be described by determinism using unknown variables?". These unknown variables are said hidden-variables.

In \cite{bell1964einstein}, Bell shows that if there exist local hidden-variables, they are necessary to satisfy a certain inequality which is called the Bell inequality. Note that "locality" means that an object is influenced directly only by its nearby surroundings. 

After all, the Bell inequality is experimentally tested and it is explored that measured values do not satisfy the inequality. (See \cite{aspect1982experimental}) As a result, quantum mechanics is not described by local hidden-variable theory. Also, other evidences for denying hidden-variable theory (no-go theorem) are given under various assumptions. For instance, Kochen-Specker \cite{kochen1975problem}, Hardy \cite{hardy1993nonlocality} and Greenberger, Horne, Zeilinger (GHZ) \cite{greenberger1989going}, and so on.

After that, Abramsky and Brandenburger \cite{abramsky2011sheaf} classify their no-go theorems. In their study, the sheaf-theoretic method plays an important part. It will be reviewed in \S $2$.

We consider about two types of locality in \S $3$. One is in sheaf-theoretic structure and the other is in algebraic probability theory. Algebraic probability \cite{hora2007quantum} has a contribution to describe quantum phenomenon. It contains classical probability, in other word, measure-theoretic probability by Kolmogorov). In \cite{abramsky2011sheaf}, locality arises from one of more general notion of factorizability (Definition \ref{4}). By contrast, in \cite{khrennikov2001foundations}, locality assumption is defined by independent probabilities in two-compound system. We extend it to $n$-compound system in Definition \ref{2}. We also show locality assumption induces the Bell inequality violation. It justifies our locality assumption as a mathematical formulation.

We give a relationship between the sheaf-theoretic structure and the algebraic probability as an equivalence of categories in Theorem \ref{3}. In \S $4$ and below sections, we use category theory \cite{awodey2010category}, because it is useful to understand their mathematical structures. We show two specific examples. First, a (quantum) state in algebraic probability space transforms to a no-signalling empirical model in sheaf-theoretic structure. Second, the factoriability transforms to the independence of probabilities. This theorem is also true in the case with representations (Corollary \ref{11}). The operators of quantum mechanics (e.g. rotation) is described by representations. So, this corollary is physically meaningful.

Next, we discuss how to extend previous discussion to algebraic quantum field theory (AQFT for short) in \S $5$ and \S $6$. We show the notion of locality which divides the sapcetime in Theorem \ref{5}. It induces the Schlieder Property \cite{halvorson2006algebraic}. When the property is satisfied, the Bell inequality in AQFT is violateed \cite{landau1987violation}. We also show the Split Property \cite{halvorson2006algebraic,ojima2015local} induces no-signalling property in Theorem \ref{8}.

Sheaf theory is known as a good tool for analyzing AQFT \cite{haag1996problem,ojima2015local}. It is called a local state. This sheaf divides the spacetime into many parts. The codomain of it is equal to the whole observables. Using equivalence of categories in Theorem \ref{3}, observables correspond to measurements of the sheaf-theoretic structure. Then, we can use the original sheaf which is defined in \cite{abramsky2011sheaf}. Therefore, the spacetime is connected with the sheaf-theoretic structure. The composite of their two sheaves is defined in Definition \ref{12}.

Finally, under some assumptions, we show the elements of the composite are no-signalling empirical models (Theorem \ref{13}). It can be regarded as No-Signalling theorem in AQFT. So, we call it a spacetime No-Signalling theorem. One of the assumptions of this theorem suggests that we can discuss a compound of systems. In \cite{mansfield2013mathematical}, the compound of systems is described by category theory, notably by a symmetric monoidal category. Hence, this theorem promotes a development of the sheaf-theoretic structure theory and the connection with AQFT.

\section{Review of the sheaf-theoretic structure}
In \cite{abramsky2011sheaf}, measurements and contexts are mathematically characterized. We review how to describe their notions. We discuss this formulation as a probability theory which is independent of physical meaning on quantum mechanics in this section.
\begin{itemize}
\item {\bf Measurement scenario} $(X,O,\mathcal{M})$ is the tuple of two sets $X,O$ and one of the antichain covers $\mathcal{M}$ of $X$: if $C,C'\in\mathcal{M},C\subset C'$ then $C=C'$ and $\bigcup_{C\in\mathcal{M}}C=X$. Note that, $X$ represents a measurement, and $O$ represents an outcome of $X$.

\item {\bf Maximal context} $C$ is an element of $\mathcal{M}$. The subset $U$ of $X$ represents a context.

\item {\bf Event sheaf} $\mathcal{E}:\mathbf{P}^{\mathrm{op}}(X)\rightarrow\mathbf{Set}$ is a presheaf where $\mathbf{P}(X)$ is the poset category of measurements $X$ and $\mathbf{Set}$ is the category of sets. Notice that, $\mathcal{E}(U)$ is the set of functions $f:U\rightarrow O$. In this paper, we don't clearly distinguish the word "presheaf" from "sheaf".

\item {\bf Distribution} is a function $d$ from $U$ to commutative semiring $R$ with finite support. It is supposed to satisfy
\[
\sum\limits_{U\subset X}d(U)=1.
\]

\item {\bf Distribution functor} is a functor which transfers the category of the sets of distributions $\mathcal{D}_R(C)$ to $\mathbf{Set}$. It is supposed that, given a morphism $f:C_1\rightarrow C_2$,
\[
\mathcal{D}_R(f):\mathcal{D}_R(C_1)\rightarrow\mathcal{D}_R(C_2)::d\mapsto[C_2\mapsto\sum\limits_{f(C_1)=C_2}d(C_1)].
\]

\item {\bf Marginal} is the function
\[
d|_{U(s)}:=\sum\limits_{\substack{s'\in\mathcal{E}(U')\\s'|_U=s}}d(s')
\]
where $s\in\mathcal{E}(U)$. It is supposed that for all $U'\supset U$,
\[
\mathcal{D}_R\mathcal{E}(U')\rightarrow\mathcal{D}_R\mathcal{E}(U)::d\mapsto d|_U
\]
where $\mathcal{D}_R\mathcal{E}:\mathbf{P}^{\mathrm{op}}(X)\rightarrow\mathbf{Set}$ is the composite of $\mathcal{E}$ and $\mathcal{D}_R$.

\item {\bf Empirical model} $e_C\in\mathcal{D}_R\mathcal{E}(C)$ is the section of event sheaf. The sequence of empirical models determines the measurement scenario.
\end{itemize}

\section{Two types of locality}
\subsection{Locality of the sheaf-theoretical structure}
Abramsky and Brandenburger formulate a theory of hidden-variables in \cite{abramsky2011sheaf}. A discussion using hidden-variables, in other word, hidden-variable theory means that when we measure physical quantities, there exists a hidden-variable which makes relationship between two measurements. However, we think their measured values behave probabilistic.

In order to define naturaly the notion of locality, we suppose the existence of a hidden-variable. Note that, this existence is denied in the discussion about the Bell inequality.

Fix the set $\Lambda$ for the values of the hidden-variables. Given $\lambda\in\Lambda$, we consider empirical model $h_C^\lambda\in\mathcal{D}_R\mathcal{E}(C)$. Suppose $h_\Lambda\in\mathcal{D}_R(\Lambda)$ is the distribution of the hidden-variables. In order to describe an empirical model by hidden variables, for every $s\in\mathcal{E}(C)$
\[
e_C(s)=\sum\limits_{\lambda\in\Lambda}h_C^\lambda(s)\cdot h_\Lambda(\lambda)
\]
is required. It can be realized that the hidden-variables exist when this condition is fulfilled. In the case that the existence of hidden-variables is assumed, the empirical model $e_C$ is statistically described as averaging. For example, in the thermodynamics, we regard particles being deterministic. Hence, it looks like a probability, but actually it is determinism with measurement errors.

The compatibility of hidden-variable theory is defined as follows:
\begin{defi}{\rm(\cite[Sect. 8]{abramsky2011sheaf})}
If the distribution of hidden-variable $h_C^\lambda$ is compatible, in other word for any two maximal contexts $C$ and $C'$
\[
h_C^\lambda|_{C\cap C'}=h_{C'}^\lambda|_{C\cap C'}
\]
is satisfied, it is said to be {\bf parameter independence}.
\end{defi}

The condition like parameter independence of the general empirical model is called no-signalling. It forbids superluminal signalling.
\begin{defi}{\rm(\cite[Sect. 2.5]{abramsky2011sheaf})}
If the empirical model $e_C$ is compatible, in other word for any two maximal contexts $C$ and $C'$ such that
\[
e_C|_{C\cap C'}=e_{C'}|_{C\cap C'},
\]
it is said to be {\bf no-signalling}.
\end{defi}

No-signaling is a more general assumption than quantum realization (all we need for the assumptions when we discuss quantum mechanics). For example, PR box \cite{popescu1994quantum} satisfies no-signaling while it does not satisfy quantum realization. We will assume no-signalling as a minimum requirement of the empirical models which we will be interested in.

One of the special conditions is important for us.
\begin{defi}{\rm(\cite[Sect 8.]{abramsky2011sheaf})}\label{4}
If for every $s\in\mathcal{E}(C)$
\[
h_C^\lambda(s)=\prod\limits_{m\in C}h_C^\lambda|_{\{m\}}(s|_{\{m\}})
\]
then the hidden-variable model $h$ is called {\bf factorizability}.
\end{defi}

Hence, the hidden-variable model is described by Cartesian product which is restricted to the elements of maximal contexts. In Bell-type measurement scenario, factorizability corresponds exactly to the Bell locality \cite{bell1964einstein}. Therefore, factorizability is realized as the notion of generalized locality. (see \cite{abramsky2011sheaf})

\subsection{Locality of the algebraic probability space}

We define a few terms of algebraic probability space \cite{hora2007quantum} needed later. An {\bf algebraic probability space} is a tuple $(\mathcal{A},\varphi)$ where $\mathcal{A}$ is a $C^*$-algebra and {\bf state} $\varphi$ is a positive linear functional which satisfies $\varphi(1)=1$. Notice that $1$ is an identity element of $\mathcal{A}$.

The algebraic probability space $(L^\infty(\Omega,\mathcal{F},\mathbf{P}),\varphi)$ arises from a classical probability space $(\Omega,\mathcal{F},\mathbf{P})$ and an expectation satisfies $\varphi(f)=\mathbf{E}(f)$ for every $f\in L^\infty(\Omega,\mathcal{F},\mathbf{P})$ where $\mathbf{E}$ represents the expectation, and $L^\infty$ represents a Lebesgue space: for all $1\leq p<\infty$, $p$ norms of any function are measurable.

Every particle which has $1/2$ spin in quantum mechanics can be up or down as a spin state. The physical system whose particles can take two states is said to be a {\bf two-level system}. Notice that, it is not complete formulation of two-level system in physical term. In general, given an $n\times n$ positive Hermitian matrix $A$, the state $\rho$ (which is regarded as a density matrix) of the $n$-level system satisfies $\varphi(A)=\mathrm{Tr}(\rho A)$. In the case of $n\geq 2$, it is non-commutative. 

Fix orthonormal basis $e_1=(0,1),e_2=(1,0)$ of two-dimensional Hilbert space $\mathbf{C}^2$, where $\mathbf{C}$ represents complex numbers. Observables are given by $2\times 2$ Hermite matrix, and the whole of them $M(2,\mathbf{C}^2)$ is an algebraic probability space. Using the Pauli matrices
\[
\sigma_0=\left(
    \begin{array}{cc}
      1 & 0 \\
      0 & 1
    \end{array}
  \right)
,\sigma_1=\left(
    \begin{array}{cc}
      0 & 1 \\
      1 & 0
    \end{array}
  \right)
,\sigma_2=\left(
    \begin{array}{cc}
      0 & -i \\
      i & 0
    \end{array}
  \right)
,
\sigma_3=\left(
    \begin{array}{cc}
      1 & 0 \\
      0 & -1
    \end{array}
  \right),
\]
any Hermite matrices are described as linear combination: for all $x_0,\ldots,x_3\in\mathbf{R}$
\[
x_0\sigma_0+x_1\sigma_1+x_2\sigma_2+x_3\sigma_3= \left(
    \begin{array}{cc}
      x_0+x_3 & x_1-ix_2 \\
      x_1+ix_2 & x_0-x_3 
    \end{array}
  \right)
\]
is satisfied, where $\mathbf{R}$ represents real numbers.

Let $a=(a_1,a_2,a_3)$ be a unit vector on $\mathbf{R}^3$ and
\[
S_a:=a_1\sigma_1+a_2\sigma_2+a_3\sigma_3.
\]
This is a Hermite matrix with eigenvalue $\pm1$. We define the unit vector $u:=(e_1\otimes e_2-e_2\otimes e_1)/\sqrt{2}$ on $\mathbf{C}^2\otimes\mathbf{C}^2$ and
\[
\varphi(A):=\Braket{u,Au},(A\in M(2,\mathbf{C})\otimes M(2,\mathbf{C}))
\]
where $\Braket{\cdot,\cdot}$ is an inner product over $\mathbf{R}$. It induces algebraic probability space $(M(2,\mathbf{C})\otimes M(2,\mathbf{C}),\varphi)$. Furthermore, we consider the correlation of two particles on two-compound system. It arises from the tensor product. We fix two variables: let $a,b\in\mathbf{R}^3$ be unit vectors.

If $S_x^{(1)}:=S_x\otimes I,S_x^{(2)}:=I\otimes S_x,(x\in\{a,b\})$, where $I$ is a unit matrix, then we obtain $\varphi(S_a^{(1)}S_b^{(2)})=-\Braket{a,b}$. Is this a constraint satisfied in measure-theoretic probability space? The answer is no.
\begin{theo}{\rm(\cite[Lem. (1)]{khrennikov2001foundations})}\label{1}
Let $a,b,c\in\mathbf{R}^3$ be unit vectors. There exist three measure-theoretic random variables $X_a,X_b,X_c$ in the set $\{-1,1\}$ with
\[
\mathbf{E}(X_xX_y)\neq-\Braket{x,y}
\]
where $x,y\in\{a,b,c\}$.
\end{theo}

Notice that the spin up and down states are represented by $-1$ and $1$ in the sentence.

This theorem means limitation of measure-theoretic expectation. That is, there exist states which are not described by measure-theoretic probability. Therefore, we have to extend an expectation $\mathbf{E}$ to a state $\varphi$ like above discussion. Also, it means violation of the Bell iniquality. The Bell inequality is the following:
\begin{theo}{\bf(Bell inequality} {\rm\cite[Cor. (1)]{khrennikov2001foundations})}
Let $X_x^{j},(x\in\{a,b,c\},j=1,2)$ be measure-theoretic random variables in the set $\{-1,1\}$. Then
\[
|\mathbf{E}(X_a^{(1)}X_b^{(2)})-\mathbf{E}(X_c^{(1)}X_b^{(2)})|\leq 1-\mathbf{E}(X_a^{(1)}X_c^{(1)})=1+\mathbf{E}(X_a^{(1)}X_c^{(2)})
\]
is satisfied.
\end{theo}

If Theorem \ref{1} is not satisfied, then the Bell inequality is described as follows:
\[
|\braket{a,b}-\braket{b,c}|\leq 1-\braket{a,c}
\]

There exists no unit vectors which satisfy the Bell inequality (See \cite{khrennikov2001foundations}). It makes a contradiction. So Theorem \ref{1} is true.

When we would like to suppose that Theorem \ref{1} is always true, we need satisfying the assumption of independent measure-theoretic random variables. It induces
\[
\mathbf{E}(X_xX_y)=\mathbf{E}(X_x)\mathbf{E}(X_y)=0
\]
\[
-\Braket{x,y}=0.
\]

As a result, this assumption induces the violation of the Bell inequality. It is satisfied while the Bell inequality is violated.

In Bell's discussion, it is sufficient to think only two random variables $X^{(1)}_x$ and $X^{(2)}_x$ when we fix the unit vector. In general, we extend the definition of locality \cite{khrennikov2001foundations} to $n$-measurements as follows:
\begin{defi}{\bf(Locality assumption)}\label{2}
When we assume that hidden-variables exist, let $x$ be a unit vector in three-dimensional Euclidean space. We consider two measure-theoretic random variables $X_x^{(1)},\ldots,X_x^{(n)}$ whose values are taken in the set $\{-1,1\}$. Moreover, let $\mathfrak{M}_1,\ldots,\mathfrak{M}_n$ be the set of all possible measurement settings. Under this assumption, the sample space $\Omega$ of measure-theoretic probability space satisfies $\Omega=\Lambda\times\mathfrak{M}_1\times\cdots\times\mathfrak{M}_n$, where $\Lambda$ is the set of hidden-variables.

We consider the functions $F_x^{(1)},\ldots,F_x^{(n)}:\Lambda\times\mathfrak{M}_1\times\cdots\times\mathfrak{M}_n\rightarrow\{-1,1\}$ which satisfy, for all $\lambda\in\Lambda,m_1\in\mathfrak{M}_1,\ldots,m_n\in\mathfrak{M}_n$,
\[
X_x^{(i)}=F_x^{(i)}(\lambda,m_i)
\]
where $x\in\{a,b,c\},i=1,\ldots,n$.
\end{defi}

Measured values $X_x^{(1)},\ldots,X_x^{(n)}$ are random variables. It looks like a probabilistic phenomenon, whereas it depends on hidden-variable $\lambda$. According to this definition, random variables $X_x^{(1)},\ldots,X_x^{(n)}$ are mutually independent when we do not suppose the existence of hidden-variable $\lambda$.

\section{Duality of two categories and locality}
Two types of locality are given in the previous section. In this section we discuss the relationship between them. Firstly, we consider the framework which includes them. We extend their notions as category theory and prove the equivalence of categories.

Let $\mathfrak{N}$ be von Neumann algebra associated with a Hilbert space $\mathcal{H}$. A no-signalling empirical model on $(X_A,O,\mathcal{M}_A)$ is written $S_A$. In the below discussion, we fix $X_A=\mathfrak{N}_A$ and $O=\mathbf{R}$. The {\bf system} $A$ \cite[Sect. 1.7]{mansfield2013mathematical} is defined by tuple $(\mathfrak{N}_A,\mathcal{M}_A,S_A)$. A morphism of the system is a map $k:(\mathfrak{N}_A,\mathcal{M}_A)\rightarrow(\mathfrak{N}_B,\mathcal{M}_B)$ such that for all $C\in\mathcal{M}_A$, $k(C)\in\downarrow\mathcal{M}_B$. Note that $\downarrow\mathcal{M}_B$ is defined by
\[
\downarrow\mathcal{M}_B:=\{U\in P(\mathfrak{N}_B): U\subset C'\in\mathcal{M}_B\}
\]
where $P(\mathfrak{N}_B)$ is the poset of $\mathfrak{N}_B$. Namely, the morphism $k$ extends to contexts which contain all maximal contexts in $\mathcal{M}_B$. Let $k^\star:S_B\rightarrow S_A$ be a reversal map of states which satisfies
\[
k^\star(e_C)(s):=\sum\limits_{\substack{s'\in\mathcal{E}(k(C))\\ s'\circ k=s}}e_{f(C)}(s')
\]
where $e\in S_B$ and $C\in\mathcal{M}_A$. In addition, the morphism $k$ satisfies $k^\star(S_B)\subset S_A$. We can identify associativity and identity of this definition, thus states (or von Neumann algebra with it) and such morphisms $k$ form a category. It is called {\bf system category} (category of systems) \cite[Sect. 1.7]{mansfield2013mathematical}, which is written $\mathbf{S}$. 
\begin{theo}\label{3}
The following categories $\mathbf{S}$ and $\mathbf{N}$ are equivalent as categories:

{\rm(1)}The system category $\mathbf{S}^\mathrm{op}$ whose objects are von Neumann algebras $\mathfrak{N}$.

{\rm(2)}The objects of the category $\mathbf{N}$ are von Neumann algebras $\mathfrak{N}$ and the morphisms are injective $*$-homomorphisms $f$ which satisfy $f(A\circ B)=f(A)\circ f(B)$ and $f(A^*)=f(A)^*$.
\end{theo}
\begin{proof}
A $*$-homomorphism of the category $\mathbf{N}$ is injective, so it preserves a point. In other word, $f:\mathfrak{N}_A\rightarrow\mathfrak{N}_B$ such that $f(a)=b$ for all $a\in\mathfrak{N}_A,b\in\mathfrak{N}_B$.

Given a new element $\ast$, the functor
\[
F:\mathbf{S}^\mathrm{op}\rightarrow\mathbf{N}::\mathfrak{N}_A\mapsto\mathfrak{N}_A\cup\{\ast\}
\]
is defined. For all morphism $f:\mathfrak{N}_A\rightarrow\mathfrak{N}_B$, this functor $F(f):\mathfrak{N}_A\cup\{\ast\}\rightarrow\mathfrak{N}_B\cup\{\ast\}$ arises from
\[
f_\ast(x)=\left\{\begin{array}{ll}
f(x) & x\in k^\star(S_B)\\
\ast & \mathrm{otherwise}\\
\end{array}\right.
\]
then $f_\ast(\ast_1)=\ast_2$ for every $\ast_1\in\mathfrak{N}_A\cup\{\ast\},\ast_2\in\mathfrak{N}_B\cup\{\ast\}$. As a result $f_\ast:\mathfrak{N}_A\cup\{\ast\}\rightarrow\mathfrak{N}_B\cup\{\ast\}$ is an injective morphism.

We also define the functor
\[
G:\mathbf{N}\rightarrow\mathbf{S}^\mathrm{op}::\mathfrak{N}_A\mapsto\mathfrak{N}_A\setminus\{a\}.
\]
The morphisms are $g:\mathfrak{N}_A\rightarrow\mathfrak{N}_B,a\in\mathfrak{N}_A,b\in\mathfrak{N}_B$ which arise from
\[
G(g):\mathfrak{N}_A\setminus g^{-1}(b)\rightarrow\mathfrak{N}_B\setminus\{b\}.
\]
So, $G(g)(x)=g(x)$ where $g(x)\neq b$. Note that, $g$ is injective by the assumption of Theorem \ref{3}. It is regarded as a set-inclusion morphism.

Note that, $G\circ F$ is the identity functor on $\mathbf{S}^\mathrm{op}$ because their functor plays a role of just adding new elements and throw them away. By contrast, $F\circ G$ is not an identity functor on $\mathbf{N}$ because $a$ is not necessary being equal $*$, hence $F\circ G$ is nothing but a natural isomorphism.
\end{proof}
\begin{example}
Let $\mathcal{P}(\mathfrak{N})$ be the set of projections of $\mathfrak{N}$. Using Gleason theorem {\rm\cite[Thm. 5.3.9]{hamhalter2003quantum}}, if the type of $\mathfrak{N}$ is $I_n,(n\neq 2,\infty)$, then there exists one-to-one correspondence between the state on $\mathcal{P}(\mathfrak{N})$ and extended probability measure on $\mathcal{P}(\mathfrak{N})$. When $\mathcal{M}$ is the set of commutative sub-algebras $\mathbf{C}(\mathfrak{N})$ of $\mathfrak{N}$, we can use Generalized No-Signalling theorem {\rm\cite[Prop. 9.2]{abramsky2011sheaf}}. Hence, probability measure on $\mathcal{P}(\mathfrak{N})$ is regarded as the empirical model $e_C$, and it satisfies no-signalling property. In conclusion, the state on algebraic probability space is regarded as empirical model on $(\mathfrak{N},\mathbf{R},\mathbf{C}(\mathfrak{N}))$.
\end{example}

Next, we consider two types of locality, Definition \ref{2} and Definition \ref{4}. We consider $F^{(1)}_x,F^{(2)}_x,\ldots,F^{(n)}_x$ of the sample space $\Omega=\Lambda\times\prod_{i=1}^n\mathfrak{M}_n$. Each measurement is mutually independent. Hence, a probability measure $\mathbb{P}$ on $\Omega$ such that for all Borel set $B_i$ on each probability space with $\mathbb{P}_i$ is defined as follows:
\[
\mathbb{P}\Bigl(\prod\limits_{i=1}^nB_i\Bigl)=\prod\limits_{i=1}^n\mathbb{P}_i(B_i).
\]

On the other hand, sheaf-theoretic discussion is the following:
\[
h^\lambda_C=\prod_{m\in C}h^\lambda_C|_{\{m\}}(s|_{\{m\}}).
\]
Similarly, $h_C^\lambda=\Lambda\times\prod_{i=1}^nh_C|_{\{m_i\}},(m_i\in C,h_C\in\mathcal{D}_R\mathcal{E}(C))$ can be regarded as a sample space. Notice that, when we fix commutative semiring $R$ is a positive real number, a distribution is a probability measure, because it is defined by $\mathbb{P}(X^{-1}(B))$ where $X$ is a random variable and $B$ is a Borel set.

\begin{example}
We define the notation $\coprod_i^nX_i:=X_1\oplus\cdots\oplus X_n$ and we also consider $n$-compound system. The event sheaf is the contravariant functor, so the product on $\mathbf{P}(X)^\mathrm{op}$ induces the coproduct on $\mathbf{P}(X)$. For many numbers of products and coproducts, there exists the natural isomorphism
\[
\prod_i^n\mathrm{Hom}_\mathbf{N}(X_i,Y)\cong\mathrm{Hom}_\mathbf{S}\Bigl(\coprod_i^nX_i,Y\Bigl)
\]
because of Theorem \ref{3}. The contravariant functor (event sheaf) gives rise to product from coproduct. In other word, it is continuous: it preserves limit. 
\end{example}

When we discuss quantum mechanics, it is useful to refer representations. For example, a rotation group $SO(n)$: a group of product whose elements are standardized $n\times n$ orthogonal matrices. Suppose $R$ is a rotation, and $\hat{R}$ is a transformation of state vector $\ket{\phi_0}$ into $\ket{\phi_1}$. An operation of rotation which operates on coordination $R_2\circ R_1$ is given by the operation of the state vector $\hat{R}_2(\ket{\phi_1})=\hat{R}_2\circ\hat{R}_1(\ket{\phi_0})$. As a result, the whole of $\hat{R}$ is the representation of $SO(n)$.

Therefore, we define the category of representations. A category which arises from endomorphisms on von Neumann algebra $\mathfrak{N}$ is denoted by $\mathrm{End}_\mathbf{N}(\mathfrak{N})$. The category $\mathbf{N}_R$ is defined as follows: Suppose $j_1,j_2\in\mathrm{End}_\mathbf{N}(\mathfrak{N})$. Let $j_1,j_2$ be the objects, and the set of morphisms is given by intertwiner: given the object $N$ of $\mathbf{N}$ and endomorphism $j_1,j_2\in\mathrm{End}_\mathbf{N}(\mathfrak{N})$, intertwiner $J\in\mathfrak{N}$ satisfies $j_1(N)\circ J=J\circ j_2(N)$. We can regard $\mathbf{N}_R$ as a category of representations.

Also, let $\mathbf{S}_R$ be the category of representations: the objects are endomorphisms of von Neumann algebra and the morphisms are intertwiners with them.

Note that, $\mathbf{N}_R\neq\mathbf{S}_R$, because an object of $\mathbf{N}_R$ maps an injective morphism to other injective morphism, while the object of $\mathbf{S}_R$ maps a set-inclusion morphism to other set-inclusion morphism.

Using Theorem \ref{3}, we obtain the following:
\begin{cor}\label{11}
Two categories $\mathbf{N}_R$ and $\mathbf{S}_R$ are equivalence of categories.
\end{cor}
\begin{proof}
Both $\mathbf{N}_R$ and $\mathbf{S}_R$ are functor category whose codomains are set categories. Using Tannaka duality, the forgetful functor $F_{\mathbf{N}}:\mathbf{N}\rightarrow\mathbf{Set},F_{\mathbf{S}}:\mathbf{S}\rightarrow\mathbf{Set}$ induces equivalence of categories
\[
\mathrm{Hom}(F_{\mathbf{N}_1},F_{\mathbf{N}_2})\cong\mathrm{Hom}_\mathbf{N}(\mathfrak{N}_1,\mathfrak{N}_2)
\]
\[
\mathrm{Hom}(F_{\mathbf{S}_1},F_{\mathbf{S}_2})\cong\mathrm{Hom}_\mathbf{S}(\mathfrak{N}_1,\mathfrak{N}_2).
\]

Recall Theorem \ref{3}, we obtain $\mathrm{Hom}_{\mathbf{N}_R}(F_{\mathbf{N}_1},F_{\mathbf{N}_2})\cong\mathrm{Hom}_{\mathbf{S}_R}(F_{\mathbf{S}_1},F_{\mathbf{S}_2})$.
\end{proof}

\section{Schlieder Property and Split Property}
Finally, we consider analogousely algebraic quantum field theory (AQFT for short) \cite{halvorson2006algebraic}. Its basic setting is the following: a vector space $\mathbf{R}^{1,3}:=\mathbf{R}^1\times \mathbf{R}^3$ equipped with a non-degenerate, symmetric bilinear form with signs $(-,+,+,+)$ is called {\bf Minkowski spacetime}. The notion $V_+:=\{x\in\mathbf{R}^{1,3}:x^2=x_0^2-(x_1+x_2+x_3)^2>0,x_0>0\}$ is said to be {\bf forward lightcone} and $\mathcal{O}:=(b+V_+)\cap(c-V_+)$ is called {\bf double cone} where $b,c\in\mathbf{R}^{1,3}$. The {\bf local net} which is the set of maps (or functors in general discussion) $\mathcal{O}\mapsto\mathcal{A}(\mathcal{O})$ from the set of double cones $\mathbf{K}$ to the set of $C^*$-algebras is one of the most key notions in AQFT. The axioms of AQFT are given:
\begin{itemize}
\item{\bf(Isotony)} Suppose $\mathcal{O}_1\subset\mathcal{O}_2$ then $\mathcal{A}(\mathcal{O}_1)\hookrightarrow\mathcal{A}(\mathcal{O}_2)$.
\item{\bf(Translation Covariance)} There exists a faithful, continuous representation $g\mapsto\alpha_g$ of the translation group in the group of automorphisms of $\mathcal{A}$, and
\[
\alpha_g:\mathcal{A}(\mathcal{O})\rightarrow\mathcal{A}(\mathcal{O}+g)
\]
for any double cone $\mathcal{O}$, and the translation $g$ of the translation group $G$. Note that, this condition can be extended to the Poincar\'{e} group.
\item{\bf(Microcausality)} Suppose $\mathcal{O}_1$ and $\mathcal{O}_2$ are {\bf spacelike separated}: $(x-y)^2<0,(x\in\mathcal{O}_1,y\in\mathcal{O}_2)$ then $[\mathcal{A}(\mathcal{O}_1),\mathcal{A}(\mathcal{O}_2)]=\{0\}$ where $[\cdot]$ is a commutator.
\item{\bf(Existence of vacuum)} 
There exists a {\bf vacuum} $\varphi_0$ which is the state of $\mathcal{A}$ and satisfies the following condition. Let $\pi_{\varphi_0}$ of $\mathcal{H}_{\varphi_0}$ be GNS representation for $\varphi_0$.

(1) $\varphi_0$ is $G$-invariant: for all $A\in\mathcal{A},g\in G$, $\varphi_0(\alpha_g(A))=\varphi_0(A)$.

If the condition (1) is satisfied, then there exists a strongly continuous representation $U_g$ of $G$ in the unitary group such that
\[
\pi_{\varphi_0}(\alpha_g(A))=U_g\pi_{\varphi_0}(A)U_g^*
\]
and a cyclic vector $\Omega_{\varphi_0}$ is equal to $U_g\Omega_{\varphi_0}$.

(2) The generator $P_\mu,(\mu=1,2,3,4)$ of $U_g$ satisfies $\mathrm{sp}(U_g)\in\overline{V_+}$ where $\mathrm{sp}(\cdot)$ represents spectrum and the $\overline{\cdot}$ represents closure.
\end{itemize}

By means of categorical notation, the local net is defined by the functor from double cones (objects) and inclusion maps (morphisms) to $C^*$-algebras (objects) and injective $*$-homomorphisms (morphisms). It means that, in order to characterize a physical system, they specify the interesting physical quantities on space-time.

Some other axioms (e.g. spectrum condition, additivity) in AQFT induce the following condition:
\begin{defi}{\rm(\cite[Def. 2.21]{halvorson2006algebraic}, \cite[Sect. 2.4]{ojima2015local})}
For any two double cones $\mathcal{O}_1$ and $\mathcal{O}_2$ such that the closure $\overline{\mathcal{O}}_1$ contained in $\mathcal{O}_2$, and if $E\in\mathcal{A}(\mathcal{O}_1)$ is a nonzero projection, then there exists an isometry $V\in\mathcal{A}(\mathcal{O}_2)$ such that $VV^*=E$. Notably $\mathcal{A}$ is von Neumann algebra. When it satisfies this condition, we say {\bf Propety B}.
\end{defi}

We suppose it in this paper. Notice that, if $\mathcal{A}(O)$ is type III algebra, then the local net $\mathcal{O}\mapsto\mathcal{A}(\mathcal{O})$ satisfies Propety B.

Let $A$ and $B$ are commuting $C^*$-subalgebras of some $C^*$ algebra $\mathcal{A}$. A {\bf Bell operator} for $(A,B)$ is defined by the elements of
\[
\mathcal{B}(A,B):=\{a_1b_1+a_1b_2+a_2b_1-a_2b_2:a_i=a_i^*\in A,b_i=b_i^*\in B,-1\leq a_i,b_1\leq 1\}.
\]
Under the analogous locality, the generalized Bell inequality is violated. (More detail is in \cite{landau1987violation}.)
\begin{theo}{\rm(\cite[Sect. 3.3]{halvorson2006algebraic})}
Let $\mathfrak{N}_1$ and $\mathfrak{N}_2$ be non-commutative von Neumann algebras on $\mathcal{H}$ such that $\mathfrak{N}_1\subset\mathfrak{N}_2'$. If $(\mathfrak{N}_1,\mathfrak{N}_2)$ satisfies Schlieder Property then there exists a state $\varphi$ which violates the Bell inequality: $\mathrm{sup}|\varphi(r)|>2,r\in\mathcal{B}(A,B)$ is satisfied.
\end{theo}

Note that, the state $\varphi$ can be reproduced by a local hidden variable model when $|\varphi(r)|\geq 1$ (See \cite{summers1987bell}). Hence, we can regard Schlieder Property in previous statement as candidate for locality in AQFT. It is defined as follows:
\begin{defi}{\rm(\cite[Def. 3.4]{halvorson2006algebraic})}
If $f\in\mathfrak{N}_1$ and $g\in\mathfrak{N}_2$ with a Hilbert space $\mathcal{H}_1$ are nonzero projections, then the projection on the closed subspace $e(\mathcal{H})\cap f(\mathcal{H})$ is also nonzero projection. This condition about $(\mathfrak{N}_1,\mathfrak{N}_2)$ is called {\bf Schlieder Property}.
\end{defi}

Schlieder Property can be defined in nonrelativistic quantum mechanics. In order to discuss Schlieder Property in AQFT (it imposes the axioms of spacetime), we need more precise assumption to be satisfied.
\begin{theo}{\rm(\cite[Prop. 3.13]{schlieder1969einige})}\label{7}
Suppose that the net $\mathcal{O}\mapsto\mathcal{A}(\mathcal{O})$ satisfies microcausality and Property B. If two double cones $\mathcal{O}_1,\mathcal{O}_2$ are strictly spacelike separated, then $(\mathcal{A}(\mathcal{O}_1),\mathcal{A}(\mathcal{O}_2))$ satisfies Schlieder Property.
\end{theo}

When we use {\bf strictly spacelike separated} for two double cones $\mathcal{O}_1$ and $\mathcal{O}_2$, there exists a neighborhood $N$ of zero, such that $\mathcal{O}_1+x$ is spacelike separated from $\mathcal{O}_2$ for all $x\in N$.

Therefore, strictly spacelike separated is the main condition in order to satisfy Schlieder property in AQFT.
\begin{theo}\label{5}
There exists the Cartesian product of double cones as topological spaces if and only if two double cones $\mathcal{O}_1$ and $\mathcal{O}_2$ are strictly spacelike separated. More generally, it can extend to $n$-double cones.
\end{theo}
\begin{proof}
Without loss of generality, we consider about the fundamental system of neighborhoods instead of the neighborhoods.

For a given sequence of topological space $(T_a,\mathcal{O}_a)_{a\in A}$ and $T:=\prod_{a\in A}T_\lambda$, there exists the most weak topology of $T$, such that a function $p_a:T\rightarrow T_a$ is continuous for each $a$. Let us fix finite elements $a_1,\ldots,a_n,\ldots$ of $A$. Then for all $O_{a_i}\in\mathcal{O}_{a_i}$, the set of subsets $B$ of $T$ which one can write as
\[
\bigcap\limits_{i=1}^n p_{a_i}^{-1}(O_{a_i})=\Bigl(\prod\limits_{a\in A\setminus\{a_1,\ldots,a_n\}}T_a\Bigl)\times\prod\limits_{i=1}^nO_{a_i}
\]
is a base of $T$. For each point $x$ of $T$,
\[
U_B(x):=\{U:x\in U,U\subset B\}
\]
is a fundamental system of neighborhoods whoes center point is $x$. The converse discussion is similar.
\end{proof}

Like factorizability in Definition \ref{4} or locality assumption in Definition \ref{2}, the condition of strictly spacelike separated is represented by Cartesian product. Notice that, this theorem is true when the time $\mathbf{R}^1$ which is the part of the Minkowsik space $\mathbf{R}^{1,3}=\mathbf{R}^1\times\mathbf{R}^3$ is fixed. (Remaining $\mathbf{R}^3$ is regarded as the space.)

We pay attention to the following property to relate the sheaf structure of AQFT to no-signalling.
\begin{defi}{\rm(\cite[Def. 3.9]{halvorson2006algebraic}, \cite[Sect. 3]{ojima2015local})}
Suppose that $\mathfrak{N}_1$ and $\mathfrak{N}_2$ are von Neumann algebras with $\mathcal{H}$ such that $\mathfrak{N}_1\subset\mathfrak{N}_2'$. If there exists a type $\mathrm{I}$ factor $\mathfrak{M}$ such that $\mathfrak{N}_1\subset\mathfrak{M}\subset\mathfrak{N}_2'$ then the pair $(\mathfrak{N}_1,\mathfrak{N}_2)$ is called {\bf Split Property}.
\end{defi}

For the sake of simplicity, we always suppose von Neumann algebras $\mathfrak{M}$ which induced by Split Property is a factor. Using the notion of finite index \cite{kosaki1998type} in subfactor theory, the following theorem is induced.
\begin{lemma}{\rm(\cite[Sect. 3.1]{kosaki1998type})}\label{10}
Let $\mathfrak{N}_1$ and $\mathfrak{N}_2$ be factors which have finite index: see {\rm\cite{kosaki1998type}}. If the pair $(\mathfrak{N}_1,\mathfrak{N}_2)$ satisfies Split Property, then $\mathfrak{N}_2$ is also type I factor.
\end{lemma}
\begin{theo}\label{8}
If the pair $(\mathfrak{N}_1,\mathfrak{N}_2)$ of factors with finite index satisfies Split Property, then the states of $\mathfrak{N}_2$ are no-signalling.
\end{theo}
\begin{proof}
By Split Property, for every $V\in\mathfrak{M},W\in\mathfrak{N}$, $V^*V$ and $W^*W$ are commutative. This commuting observables (in other word, commuting self adjoint operators) induce the family of distribution $\{\rho_C\}$ by way of Gleason theorem {\rm\cite[Thm. 5.3.9]{hamhalter2003quantum}} $\rho_C=\mathrm{Tr}(\rho P)$ where $\rho:\mathfrak{N}_2\rightarrow\mathbf{C}$ is the state, and $P$ is the projection of $\mathfrak{N}_2$. According to Generalized No-Signalling theorem {\rm\cite[Prop. 9.2]{abramsky2011sheaf}}, the families of distributions $\rho_C$ on families commuting observables are no-signalling. Note that, Gleason theorem is true on type $\mathrm{I}_n,(n\neq 2,\infty)$ von Neumann algebra $\mathfrak{N}_2$. This assumption is always satisfied, because of Split Property and the assumptions of Lemma \ref{10}.
\end{proof}

Thanks to this theorem, Split Property induces no-signalling condition in AQFT.
\section{From spacetime to the empirical model}
We will consider about the empirical model in AQFT.

On the discussion of restricting $C^*$-algebras to von Neumann algebra, if $\mathbf{N}$ satisfies the axiom of AQFT, then it can be regarded as the codomain of a local net. In order to consider $\mathbf{N}$ in AQFT, we have to also consider about spacetime. 

Remember the Split Property. Its necessary and sufficient conditions are the following.
\begin{theo}{\rm(\cite[Thm. 3.1]{ojima2015local})}\label{6}
Let $\pi_{\varphi_0}$ be a vacuum representation. The following conditions are equivalent.

{\rm(1)} The pair $(\mathcal{A}(\mathcal{O}_1)),\mathcal{A}(\mathcal{O}_2))$ in the local net $\{\pi_{\varphi_0}(\mathcal{A}(\mathcal{O})'')\}_{\mathcal{O}\in\mathbf{K}}$ satisfies Split Property.

{\rm(2)} A local net $\pi_{\varphi_0}(\mathcal{A}(\mathcal{O})'')$ satisfies Split Property if and only if the following condition is satisfied. For every normal state $\varphi\in\pi_{\varphi_0}(\mathcal{A}(\mathcal{O})'')$, there exists a unitary complete positive map $T$ which satisfies $T(X)=\sum_jC_j^*XC_j$ and $T(X)=\varphi(X)I,(X\in\pi_{\varphi_0}(\mathcal{A}(\mathcal{O}_1))'')$ where $C_j\in\pi_{\varphi_0}(\mathcal{A}(\mathcal{O}_2)''$.
\end{theo}

Notice that, $\pi_{\varphi_0}$ is the faithful representation, so $\mathcal{A}(O)$ satisfies Split Property when Theorem \ref{6} is true.

When we consider about local net, it is useful to focus our attention to normal states on $\mathcal{A}(\mathcal{O})$. The normal states are decided by the set of sub-localnets on bounded spacetime which we are interested in. Hence, the idea of gluing sub-local net together is natural. In \cite{haag1996problem}, they treat a spacetime point like a germ in the sheaf structure of local nets. It can contribute to the definition of thermodynamic non-equilibrium states \cite{buchholz2002thermodynamic}.

A unitary complete positive map $T$ on $\mathcal{A}$ is called {\bf local states} \cite{ojima2015local} when $T(AB)=T(A)B$ is satisfied for all $A\in\mathcal{A},B\in\mathcal{A}((\mathcal{O}_2)')$ and $T(X)=\varphi(X)I$ is satisfied for all $X\in\mathcal{A}(\mathcal{O}_1)$ where $\varphi$ is a normal state of $\mathcal{A}(O_1)$.

By Theorem \ref{6}, a local net $\mathcal{O}\mapsto\mathcal{A}(\mathcal{O})$ is divided into the local state $E_\mathbf{N}:\mathcal{O}\mapsto\mathcal{R}^{\mathcal{A}(\mathcal{O})}_\mathbf{N}\subset\mathcal{A}(\mathcal{O})$, so we can regard it as a sheaf, and also there exists a sheaf $E_\mathbf{S}:\mathcal{O}\mapsto\mathcal{R}^{\mathcal{A}(\mathcal{O})}_\mathbf{S}$ which is the counterpart of $E_\mathbf{N}$. In this sense, $\mathcal{R}^{\mathcal{A}(\mathcal{O})}_\mathbf{S}$ is equal to $\mathbf{S}^\mathrm{op}$ because of duality between $\mathbf{N}$ and $\mathbf{S}^\mathrm{op}$. This sheaf $E_\mathbf{S}$ induces other sheaf $\mathcal{D}_R\mathcal{E}(C):\mathbf{P}(\mathfrak{N})^\mathrm{op}\rightarrow\mathbf{Set}$ when we fix $X=\mathfrak{N}$. Notice that, $\mathbf{S}$ and $\mathbf{P}(\mathfrak{N})$ are equivalence of categories, because they have the same objects (von Neumann algebra) and the same morphims (set-inclusion map).

To sum up the above discussion, under Split Property (it induces no-signalling property), local states induce the sheaf $\mathcal{D}_R\mathcal{E}$. In conclusion, we define the sheaf $\mathfrak{E}$ when we fix the context $C$. 
\begin{defi}\label{12}
$\mathfrak{E}:\mathcal{D}_R\mathcal{E}\circ E_\mathbf{S}::\mathcal{O}\mapsto e_C$.
\end{defi}

Due to the discussion in section $5$, we show counterpart of Generalized No-Signalling theorem in AQFT.
\begin{lemma}{\rm(\cite[Prop. 1, Prop. 13]{redei2010quantum})}\label{9}
Suppose two von Neumann algebras $\mathfrak{N}_1$ and $\mathfrak{N}_2$ which act on a Hilbert space $\mathcal{H}$ are commuting and there exists a unitary operator $U:\mathcal{H}\rightarrow\mathcal{H}\otimes\mathcal{H}$ such that $UXYU^*=X\otimes Y,(X\in\mathfrak{N}_1,Y\in\mathfrak{N}_1)$. Then the pair $(\mathfrak{N}_1,\mathfrak{N}_2)$ satisfies Schlieder Property if and only if it satisfies Split Property.
\end{lemma}
\begin{theo}{\bf(Spacetime No-Signalling theorem)}\label{13}
If $\mathfrak{N}_1$ and $\mathfrak{N}_2$ are factors with finite index and there exists a unitary operator $U:\mathcal{H}\rightarrow\mathcal{H}\otimes\mathcal{H}$ such that $UXYU^*=X\otimes Y,(X\in\mathfrak{N}_1,Y\in\mathfrak{N}_1)$, then the elements of $\mathfrak{E}$ are no-signalling empirical model.
\end{theo}
\begin{proof}
According to Microcausality in the axioms of AQFT, if two spacetimes $\mathcal{O}_1$ and $\mathcal{O}_2$ are spacelike separable, then the localnets of them $\mathcal{A}(\mathcal{O}_1)$ and $\mathcal{A}(\mathcal{O}_2)$ are commutative. It induces Schlieder Property by Theorem \ref{7}. When the pair $(\mathfrak{N}_1,\mathfrak{N}_2)$ satisfies the assumption of this theorem, it is equivalent to Split Property by Lemma \ref{9}. Using Theorem \ref{8}, the states of $\mathfrak{N}_2$ are no-signalling. Note the fact, that Theorem \ref{6} shows the existence of local states, is guaranteed by Split Property. So, we can divide $\mathcal{A}(\mathcal{O}_1)$ and $\mathcal{A}(\mathcal{O}_2)$ to many parts by $\mathfrak{E}$, and the elements 
of each part are no-signalling.
\end{proof}

The existence of unitary operator $U:\mathcal{H}\rightarrow\mathcal{H}\otimes\mathcal{H}$ such that $UXYU^*=X\otimes Y,(X\in\mathfrak{N}_1,Y\in\mathfrak{N}_1)$ is a natural assumption, because it guarantees the existence of a compound system. In \cite[Sect. 1.7]{mansfield2013mathematical}, the compound system of two systems are defined, and it extends to the symmetric monoidal category of them.

\section*{Acknowledgments}
I extremely grateful to Tsukasa Yumibayashi for useful discussions and Jun Ueki for helpful comments. I also thank audiences at Spring Meeting of The Physical Society of Japan, Waseda University, March 2015, and Mathematical Physics Seminar in The Open University of Japan, July 2015.
\bibliography{reference1.bib}
\end{document}